\documentclass[11pt]{article}

\usepackage{color}
\usepackage{amssymb}
\usepackage[ruled,linesnumbered,algo2e,vlined]{algorithm2e}
\usepackage{amsmath}
\usepackage{microtype}
\usepackage{hyperref}
\usepackage{cleveref}
\usepackage{wrapfig}

\newcommand{\etal}{et al.}

\newcommand{\STree}{\mathsf{STree}}
\newcommand{\STrie}{\mathsf{STrie}}
\newcommand{\LST}{\mathsf{LST}}
\newcommand{\PLST}{\mathsf{preLST}}

\newcommand{\activeNode}{\mathit{activePoint}}

\newcommand{\nextNode}{\mathit{insertPoint}}

\newcommand{\Root}{\mathit{root}}
\newcommand{\newLeaf}{\mathit{newLeaf}}
\newcommand{\PrevLeaf}{\mathit{prevLeaf}}
\newcommand{\PrevInsPoint}{\mathit{prevInsPoint}}
\newcommand{\PrevLabel}{\mathit{prevLabel}}
\newcommand{\Flag}{\mathit{mismatch}}
\newcommand{\LastNode}{\mathit{newNode}}

\newcommand{\Rlink}{\mathsf{rlink}}
\newcommand{\FLink}{\mathsf{fastLink}}
\newcommand{\Slink}{\mathsf{slink}}

\newcommand{\Tree}{\mathcal{T}}
\newcommand{\PTree}{\mathcal{P}}

\newcommand{\STriedepth}{\mathsf{STrieDepth}}
\newcommand{\Depth}{\mathsf{depth}}

\newcommand{\Child}{\mathsf{child}}
\newcommand{\TChild}{\mathsf{t1child}}
\newcommand{\Parent}{\mathsf{parent}}
\newcommand{\TParent}{\mathsf{t1parent}}

\newcommand{\CreateTypeTwo}{\mathsf{createType2}}

\newcommand{\FastMatching}{\mathsf{fastMatching}}
\newcommand{\FastDecompact}{\mathsf{fastDecompact}}
\newcommand{\ReadEdge}{\mathsf{readEdge}}
\newcommand{\Split}{\mathsf{split}}

\newcommand{\NULL}{\mathsf{NULL}}
\newcommand{\Plus}{\mathsf{+}}
\newcommand{\Type}{\mathsf{type}}
\newcommand{\Label}{\mathsf{label}}

\newcommand{\True}{\mathsf{true}}
\newcommand{\False}{\mathsf{false}}

\newcommand{\ot}{:=}

\newcommand{\Ret}{\mathbf{return}}

\newcommand{\Leaf}{\mathit{leaf}}

\SetKwProg{Fn}{Function}{}{end}

\usepackage[top=2cm, bottom=2cm, left=2cm, right=2cm, includefoot]{geometry}
\usepackage{authblk}
\usepackage{graphicx}
\usepackage{amsthm}
\usepackage{url}

\newtheorem{definition}{Definition}
\newtheorem{theorem}{Theorem}
\newtheorem{lemma}{Lemma}

\begin{document}

\title{Online Algorithms for Constructing Linear-size Suffix Trie}

\author[1]{Diptarama Hendrian}
\author[2]{Takuya Takagi}
\author[3]{Shunsuke Inenaga}

\affil[1]{Graduate School of Information Sciences, Tohoku University, Japan\\
	\texttt{diptarama@tohoku.ac.jp}}
\affil[2]{Fujitsu Laboratories Ltd., Japan}
\affil[3]{Department of Informatics, Kyushu University, Japan}

\date{}

\maketitle            

\begin{abstract}
The suffix trees are fundamental data structures for various kinds of string processing.
The suffix tree of a string $T$ of length $n$ has $O(n)$ nodes and edges,
and the string label of each edge is encoded by a pair of positions in $T$.
Thus, even after the tree is built, the input text $T$ needs to be kept stored
and random access to $T$ is still needed. 
The \emph{linear-size suffix tries} (\emph{LSTs}), proposed by Crochemore et al.
[Linear-size suffix tries, TCS 638:171-178, 2016],
are a ``stand-alone'' alternative to the suffix trees.
Namely, the LST of a string $T$ of length $n$ occupies $O(n)$ total space,
and supports pattern matching and other tasks in the same efficiency as the suffix tree
without the need to store the input text $T$.
Crochemore et al. proposed an \emph{offline} algorithm which transforms
the suffix tree of $T$ into the LST of $T$ in $O(n \log \sigma)$ time and $O(n)$ space,
where $\sigma$ is the alphabet size.
In this paper, we present two types of \emph{online} algorithms
which ``directly'' construct the LST, from right to left, and from left to right,
without constructing the suffix tree as an intermediate structure.
Both algorithms construct the LST incrementally when a new symbol is read,
and do not access to the previously read symbols.
The right-to-left construction algorithm works in $O(n \log \sigma)$ time and $O(n)$ space
and the left-to-right construction algorithm works in $O(n (\log \sigma + \log n / \log \log n))$ time and $O(n)$ space.
The main feature of our algorithms is that the input text does not need to be stored.

\end{abstract}

\section{Introduction}
Suffix tries are conceptually important string data structures
that are the basis of more efficient data structures.
While the suffix trie of a string $T$ supports fast queries and operations
such as pattern matching, the size of the suffix trie can be $\Theta(n^2)$ in the worst case,
where $n$ is the length of $T$.
By suitably modifying suffix tries, 
we can obtain linear $O(n)$-size string data structures 
such as suffix trees~\cite{Weiner1973},
suffix arrays~\cite{Manber1993},
directed acyclic word graphs (DAWGs)~\cite{Blumer1985},
compact DAWGs (CDAWGs)~\cite{Blumer1987},
position heaps~\cite{Ehrenfeucht2011}, and so on.
In the case of the integer alphabet of size polynomial in $n$,
all these data structures can be constructed in $O(n)$ time and space
in an \emph{offline} manner~\cite{Crochemore1997const,Crochemore1997,Farach-ColtonFM00,FujishigeTIBT16,INIBT16,KarkkainenSB06,NarisawaHIBT17}.
In the case of a general ordered alphabet of size $\sigma$,
there are \emph{left-to-right} \emph{online} construction algorithms 
for suffix trees~\cite{Ukkonen1995},
DAWGs~\cite{Blumer1985}, CDAWGs~\cite{InenagaHSTAMP05}, and position heaps~\cite{Kucherov2013}.
Also, there are \emph{right-to-left} \emph{online} construction algorithms 
for suffix trees~\cite{Weiner1973} and position heaps~\cite{Ehrenfeucht2011}.
All these online construction algorithms run in $O(n \log \sigma)$ time with $O(n)$ space.

Suffix trees are one of the most extensively studied string data structures,
due to their versatility.
The main drawback is, however, that
each edge label of suffix trees needs to be encoded as a pair of text positions,
and thus the input string needs to be kept stored and be accessed even after
the tree has been constructed.
Crochemore et al.~\cite{Crochemore2016} proposed a new suffix-trie
based data structure called \emph{linear-size suffix tries} (\emph{LSTs}).
The LST of $T$ consists of the nodes of the suffix tree of $T$,
plus a linear-number of auxiliary nodes and suffix links.
Each edge label of LSTs is a single character,
and hence the input text string can be discarded after the LST has been built.
The total size of LSTs is linear in the input text length,
yet LSTs support fundamental string processing queries such as pattern matching
within the same efficiency as their suffix tree counterpart~\cite{Crochemore2016}.

Crochemore et al.~\cite{Crochemore2016} showed
an algorithm which transforms the \emph{given} suffix tree of string $T$
into the LST of $T$ in $O(n \log \sigma)$ time and $O(n)$ space.
This algorithm is \emph{offline}, since it requires the suffix tree to be completely built first.
No efficient algorithms which construct LSTs \emph{directly} (i.e. without suffix trees)
and in an \emph{online} manner were known.

This paper proposes two online algorithms that construct LSTs directly from the given string.
The first algorithm is based on Weiner's suffix tree construction~\cite{Weiner1973},
and constructs the LST of $T$ by scanning $T$ \emph{from right to left}.
On the other hand, the second algorithm is based on Ukkonen's suffix tree construction~\cite{Ukkonen1995},
and constructs the LST of $T$ by scanning $T$ from \emph{left to right}.
Both algorithms construct the LST incrementally when a new symbol is read,
and do not access the previously read symbols.
This also means that our construction algorithms do not need to store the input text,
and the currently processed symbol in the text can be immediately discarded
as soon as the symbol at the next position is read.
The right-to-left construction algorithm works in $O(n \log \sigma)$ time and $O(n)$ space
and the left-to-right construction algorithm works in $O(n (\log \sigma + \frac{\log n}{ \log \log n}))$ time and $O(n)$ space.

\section{Preliminaries}

Let $\Sigma$ denote an \emph{alphabet} of size $\sigma$.
An element of $\Sigma^*$ is called a \emph{string}.
For a string $T \in \Sigma^*$, the length of $T$ is denoted by $|T|$.
The \emph{empty string}, denoted by $\varepsilon$, is the string of length $0$.
For a string $T$ of length $n$, $T[i]$ denotes the $i$-th symbol of $T$
and $T[i:j] = T[i]T[i+1] \dots T[j]$ denotes the \emph{substring} of $T$ that begins at position $i$ and ends at position $j$ for $1 \leq i \leq j \leq n$.
Moreover, let $T[i:j] = \varepsilon$ if $i > j$.
For convenience, we abbreviate $T[1:i]$ to $T[:i]$ and $T[i:n]$ to $T[i:]$,
which are called \emph{prefix} and \emph{suffix} of $T$, respectively.

\subsection{Linear-size suffix trie}
The \emph{suffix trie} $\STrie(T)$ of a string $T$ is a trie
that represents all suffixes of $T$.
The \emph{suffix link} of each node $U$ in $\STrie(T)$
is an auxiliary link that points to $V = U[2:|U|]$. 
The \emph{suffix tree}~\cite{Weiner1973} $\STree(T)$ of $T$ is a path-compressed trie
that represents all suffixes of $T$.
We consider the version of suffix trees where
the suffixes that occur twice or more in $T$
can be represented by non-branching nodes.
The \emph{linear-size suffix trie} $\LST(T)$ of a string $T$,
proposed by Crochemore et al.~\cite{Crochemore2016},
is another kind of tree that represents all suffixes of $T$,
where each edge is labeled by a single symbol.
The nodes of $\LST(T)$ are a subset of the nodes of $\STrie(T)$,
consisting of the two following types of nodes:
\begin{enumerate}
	\item Type-1: The nodes of $\STrie(T)$ whose that also nodes of $\STree(T)$.
	\item Type-2: The nodes of $\STrie(T)$
              that not type-1 nodes and their suffix links point to type-1 nodes.
\end{enumerate}
A non-suffix type-1 node has two or more children and a type-2 node has only one child.
When $T$ ends with a unique terminate symbol $\$$ that
does not occur elsewhere in $T$,
then all type-1 nodes in $\LST(T)$ has two or more children.
The nodes of $\STrie(T)$ that are neither type-1 nor type-2 nodes of $\LST(T)$
are called \emph{implicit nodes} in $\LST(T)$.

We identify each node in $\LST(T)$ by the substring of $T$ that is the path label from $\Root$ to the node in $\STrie(T)$.
Let $U$ and $V$ be nodes of $\LST(T)$ such that $V$ is a child of $U$.
The edge label of $(U,V) = c$ is the same as the label of the first edge on the path from $U$ to $V$ in $\STrie(T)$.
If $V$ is not a child of $U$ in $\STrie(T)$, i.e. the length of the path label from $U$ to $V$ is more than one,
we put the $\Plus$ sign on $V$ and we call $V$ a $\Plus$-node.
\Cref{fig:suffix_trie_tree} shows an example of a suffix trie, linear-size suffix trie, and suffix tree.

For convenience, we assume that there is an auxiliary node $\bot$
as the parent of the root of $LST(T)$,
and that the edge from $\bot$ to the root is labeled by any symbol.
This assures that for each symbol appearing in $T$
the root has a non $\Plus$ child.
This will be important for the construction of LSTs 
and pattern matching with LSTs (c.f. Lemma~\ref{lem:readlabel}).

\begin{figure}[t]
	\centering
	\begin{minipage}[t]{0.32\hsize}
		\centering
		\includegraphics[scale=1.1]{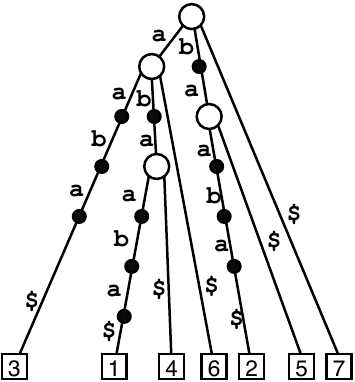}\\
		\ \ \ \small{Suffix trie}
	\end{minipage}
	\begin{minipage}[t]{0.32\hsize}
		\centering
		\includegraphics[scale=1.1]{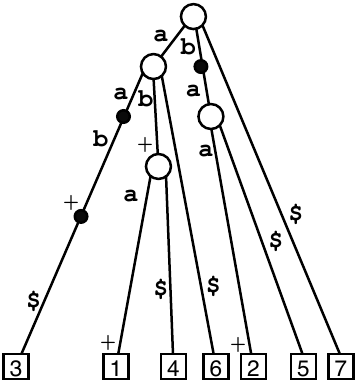}\\
		\ \ \ \small{Linear-size suffix trie}
	\end{minipage}
	\begin{minipage}[t]{0.32\hsize}
		\centering
		\includegraphics[scale=1.1]{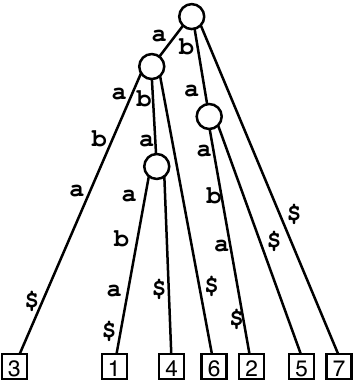}\\
		\ \ \ \small{Suffix tree}
	\end{minipage}
	\caption{
		The suffix trie, linear-size suffix trie, and suffix tree of $T = \mathtt{abaaba\texttt{\$}}$.
	}
	\label{fig:suffix_trie_tree}
\end{figure}

In the description of our algorithms, we will use the following notations.
For any node $U$, $\Parent(U)$ denotes the parent node of $U$.
For any edge $(U, V)$,
$\Label(U,V)$ denotes the label of the edge connecting $U$ and $V$,
For a node $U$ and symbol $c$,
$\Child(U,c)$ denotes the child of $U$ whose incoming edge label is $c$,
if it exists.
We denote $\Plus(U)=\True$ if $U$ is a $\Plus$-node, and $\Plus(U)=\False$ otherwise.
The suffix link of a node $U$ is defined as $\Slink(U)=V$, where $V = U[2:|U|]$.
The reversed suffix link of a node $U$ with a symbol $c \in \Sigma$
is defined as $\Rlink(U,c)=V$, if there is a node $U$ such that $cU = V$.
It is undefined otherwise.
For any type-1 node $U$,
$\TParent(U)$ denotes the nearest type-1 ancestor of $U$,
and $\TChild(U,c)$ denotes the nearest type-1 descendant of $U$ on $c$ edge.
For any type-2 node $U$,
$\Child(U)$ is the child of $U$,
and $\Label(U)$ is the label of the edge connecting $U$ and its child.

\subsection{Pattern matching using linear-size suffix trie}
In order to efficiently perform pattern matching on LSTs,
Crochemore \etal{}~\cite{Crochemore2016} introduced \emph{fast links} that are
a chain of \emph{suffix links of edges}.
\begin{definition}
	For any edge $(U,V)$, let $\FLink(U,V)=(\Slink^h(U),\Slink^h(V))$
	such that $\Slink^{h}(U) \ne \Parent(\Slink^{h}(V))$ and $\Slink^{h-1}(U) = \Parent(\Slink^{h-1}(V))$, where $\Slink^0(U) = U$ and  $\Slink^i(U) = \Slink(\Slink^{i-1}(U))$.
\end{definition}
Here, $h$ is the minimum number of suffix links that we need to  
traverse so that $\Slink^{h}(U) \ne \Parent(\Slink^{h}(V))$.
Namely, after taking $h$ suffix links from edge $(U,V)$,
there is at least one type-2 node in the path from $\Slink^h(U)$ to $\Slink^h(V)$.
Since type-2 nodes are not branching,
we can use the labels of the type-2 nodes in this path to retrieve
the label of the edge $(U, V)$ (see Lemma~\ref{lem:readlabel} below).
Provided that $\LST(T)$ has been constructed,
the fast link $\FLink(U,V)$ for every edge $(U,V)$ can be computed
in a total of $O(n)$ time and space~\cite{Crochemore2016}.

\begin{lemma}[\cite{Crochemore2016}]\label{lem:readlabel}
  The underlying label of a given edge $(U,V)$ of length $\ell$ can be
  retrieved in $O(\ell\log \sigma)$ time by using fast links.
\end{lemma}

Crochemore et al.~\cite{Crochemore2016} claimed that
due to Lemma~\ref{lem:readlabel}
one can perform pattern matching for a given pattern $P$ in $O(|P| \log \sigma)$ time
with the LST.
However, the proofs provided in~\cite{Crochemore2016}
for the correctness and time efficiency of their pattern
matching algorithm looks unsatisfactory to us,
because the algorithm of Crochemore et al.~\cite{Crochemore2016} does not seem
to guarantee that the label of a given edge is retrieved
sequentially from the first symbol to the last one (see also~\cite{Takagi2017}).
Still, in the following lemma we present
an algorithm which efficiently performs the longest prefix match
for a given pattern on the LST with fast links:
\begin{lemma}\label{lem:patmatch}
	Given $\LST(T)$ and a pattern $P$, 
	we can find the longest prefix $P'$ of $P$ that occurs in $T$ in $O(|P'|\log\sigma)$ time.
\end{lemma}

\begin{figure}[t]
        \centering
	\includegraphics[scale=0.5]{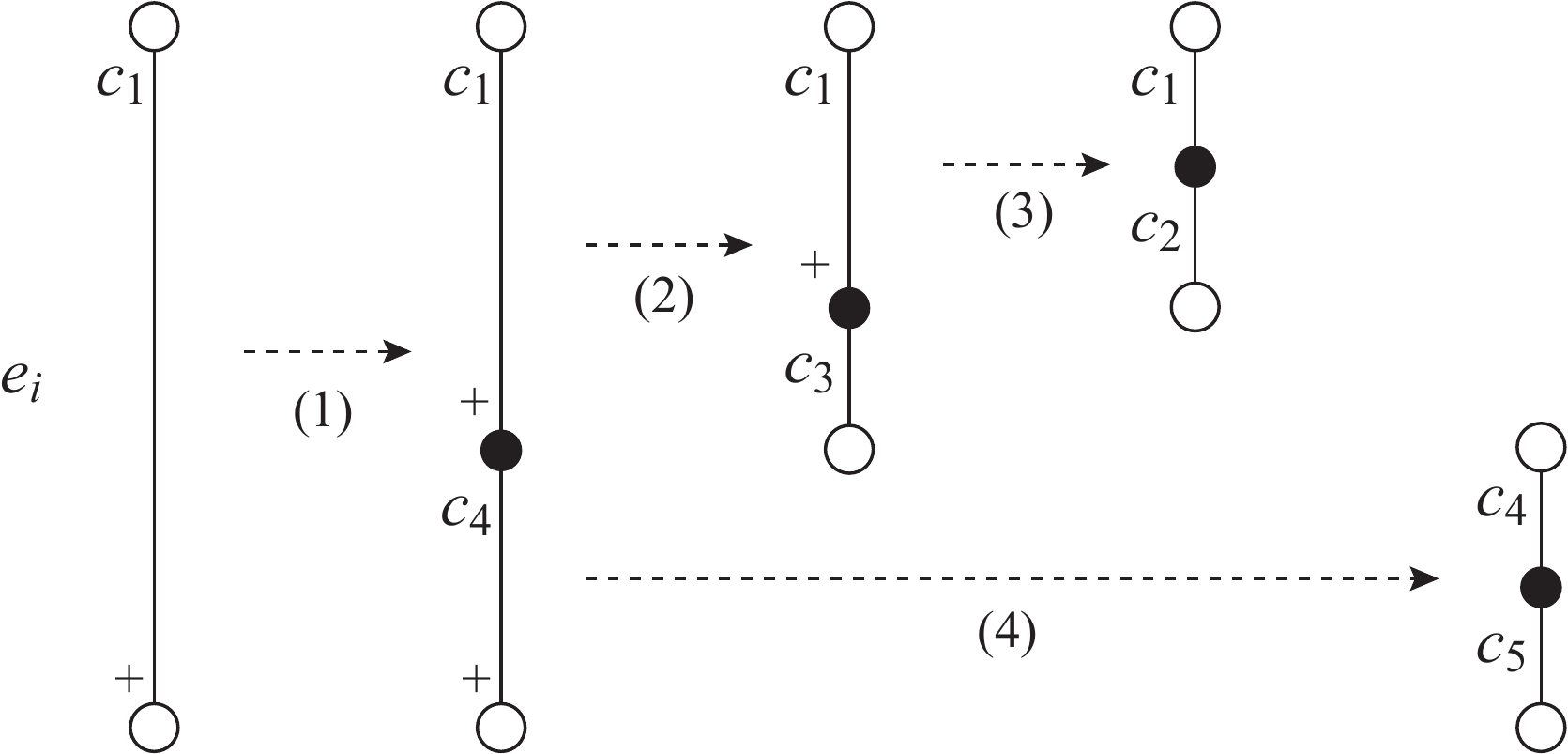}
	\caption{Illustration for our pattern matching algorithm with LST. The dashed arrows represent fast links. The number in parentheses show the orders of applications of fast links when traversing $P_i = c_1 c_2 c_3 c_4 c_5$ on the edge $e_i$.}
	\label{fig:pattern_matching}
\end{figure}

\begin{proof}
  Let $P_{1}P_{2} \cdots P_{m} = P'$ be the factorization of $P'$ such that $P_{1} \cdots P_{i}$ is a node in $\LST(T)$ for $1 \le i < m$,  $P_{1} \cdots P_{i} = \Parent(P_{1} \cdots P_{i+1})$ for $1 \le i < m-1$, and $P_{1} \cdots P_{m-1}$ is the longest prefix of $P'$ that is a node in $\LST(T)$. If $P_{1} \cdots P_{m-1} = P'$, then $P_m = \varepsilon$. In what follows, we consider a general case where $P_m \neq \varepsilon$.

  Suppose we have successfully traversed up to $P_1 \cdots P_{i-1}$,
  and let $U$ be the node representing $P_1 \cdots P_{i-1}$.
  If $U$ has no out-going edge labeled $c_1 = P_i[1] = P[|P_1 \cdots P_{i-1}|+1]$,
  then the traversal terminates on $U$.
  Suppose $U$ has an out-going edge labeled $c_1$ and let $V$ be the child
  of $U$ with the $c_1$-edge.
  We denote this edge by $e_i = (U, V)$.
  See also Figure~\ref{fig:pattern_matching} for illustration.
  If $V$ is a not $\Plus$-node, then we have read $c_1$
  and set $U \leftarrow V$ and continue with the next symbol $c_2 = P_i[2] = P[|P_1 \cdots P_{i-1}|+2]$.
  Otherwise (if $V$ is a $\Plus$-node),
  then we apply $\FLink$ from edge $(U, V)$ recursively,
  until reaching the edge $(U', V')$ such that $V'$ is not a $\Plus$-node.
  Then we move onto $V'$. Note that by the definition of $\FLink$,
  $V'$ is always a type-2 node.
  We then continue the same procedure by setting $U \leftarrow V'$
  with the next pattern symbol $c_2$.
  This will be continued until we arrive at the first edge $(U, V)$
  such that $V$ is a type-1 node.
  Then, we trace back the chain of $\FLink$'s from $(U, V)$
  until getting back to the type-2 node $V''$
  whose out-going edge has the next symbol to retrieve.
  We set $U \leftarrow V''$ and continue with the next symbol.
  This will be continued until we traverse all symbols $c_j$ in $P_i$
  for increasing $j = 1, \ldots, |P_i|$ along the edge $e_i$,
  or find the first mismatching symbols.

  The correctness of the above algorithm follows from the fact that
  every symbol in label of the edge $e_i$ is retrieved from a type-2 node
  that is not branching, except for the first one retrieved from
  the type-1 node that is the origin of $e_i$.
  Since any type-2 node is not branching,
  we can traverse the edge $e_i$ with $P_i$
  iff the underlying label of $e_i$ is equal to $P_i$ for $1 \leq i \leq m-1$.
  The case of the last edge $e_{m}$ where the first mismatching symbols are found is analogous.
    
  To analyze the time complexity, we consider the number of applications of $\FLink$.
  For each $1 \leq i \leq m-1$, the number of applications of $\FLink$ is
  bounded by the length of the underlying label of edge $e_i$, which is $|P_i|$.
  This is because each time we follow a $\FLink$, at least one new symbol is retrieved.
  Hence we can traverse $P_1 \cdots P_{m-1}$ in $O(|P_1 \cdots P_{m-1}| \log \sigma)$ time.
  For the last fragment $P_m$,
  we consider the number of applications of $\FLink$ until we find
  the type-2 node $X$ whose out-going edge has the first mismatching symbol.
  Since the first application of $\FLink$ for $P_m$
  begins with an edge whose destination has string depth $|P_1 \cdots P_{m-1}|$
  and since each symbol appearing in $T$ is represented by a node as a child of the root,
  the number of applications of $\FLink$ until finding $X$ is bounded by $|P_1 \cdots P_{m-1}|$.
  Note that this is independent of the length of the edge $e_m$
  which can be much longer than $P_m$.
  After finding $X$, we can traverse $P_m$ as in the same way to previous $P_i$'s.
  Thus, we can traverse $P_m$ in $O(|P_1 \cdots P_{m}| \log \sigma)$.
  Overall, it takes $O(|P_1 \cdots P_m| \log \sigma)$ time
  to traverse $P' = P_1 \cdots P_m$. This completes the proof.
\end{proof}
Algorithm~\ref{alg:patmatch} in Appendix
shows a pseudo-code of our pattern matching algorithm with the LST in Lemma~\ref{lem:patmatch}.

\section{Right-to-left online algorithm}

In this section, we present an online algorithm that constructs $\LST(T)$
by reading $T$ from right to left.
Let $\Tree_{i} = \LST(T[i:])$ for $1 \leq i \leq n$.
Our algorithm constructs $\Tree_{i}$ from $\Tree_{i+1}$ incrementally when $c = T[i]$ is read.
For simplicity, we assume that $T$ ends with a unique terminal symbol $\texttt{\$}$
such that $T[i] \ne \texttt{\$}$ for $1 \le i < n$.

We remark that the algorithm does not construct fast links of the LSTs.
The fast links can easily be constructed in $O(n)$ time after $\LST(T)$ has been constructed.

\begin{figure}[t]
        \centering
	\fbox{\includegraphics[scale=0.4]{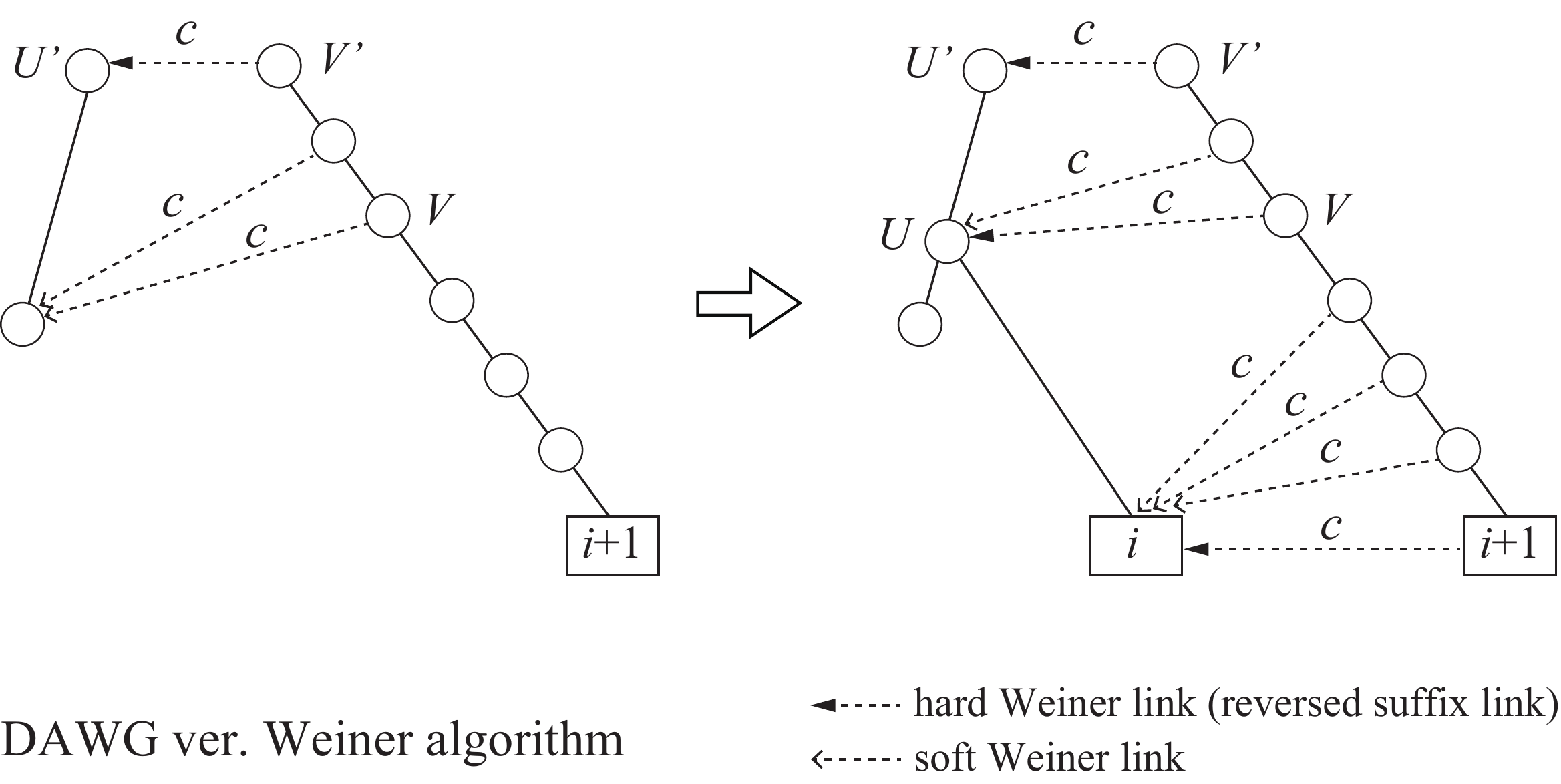}}
        \fbox{\includegraphics[scale=0.4]{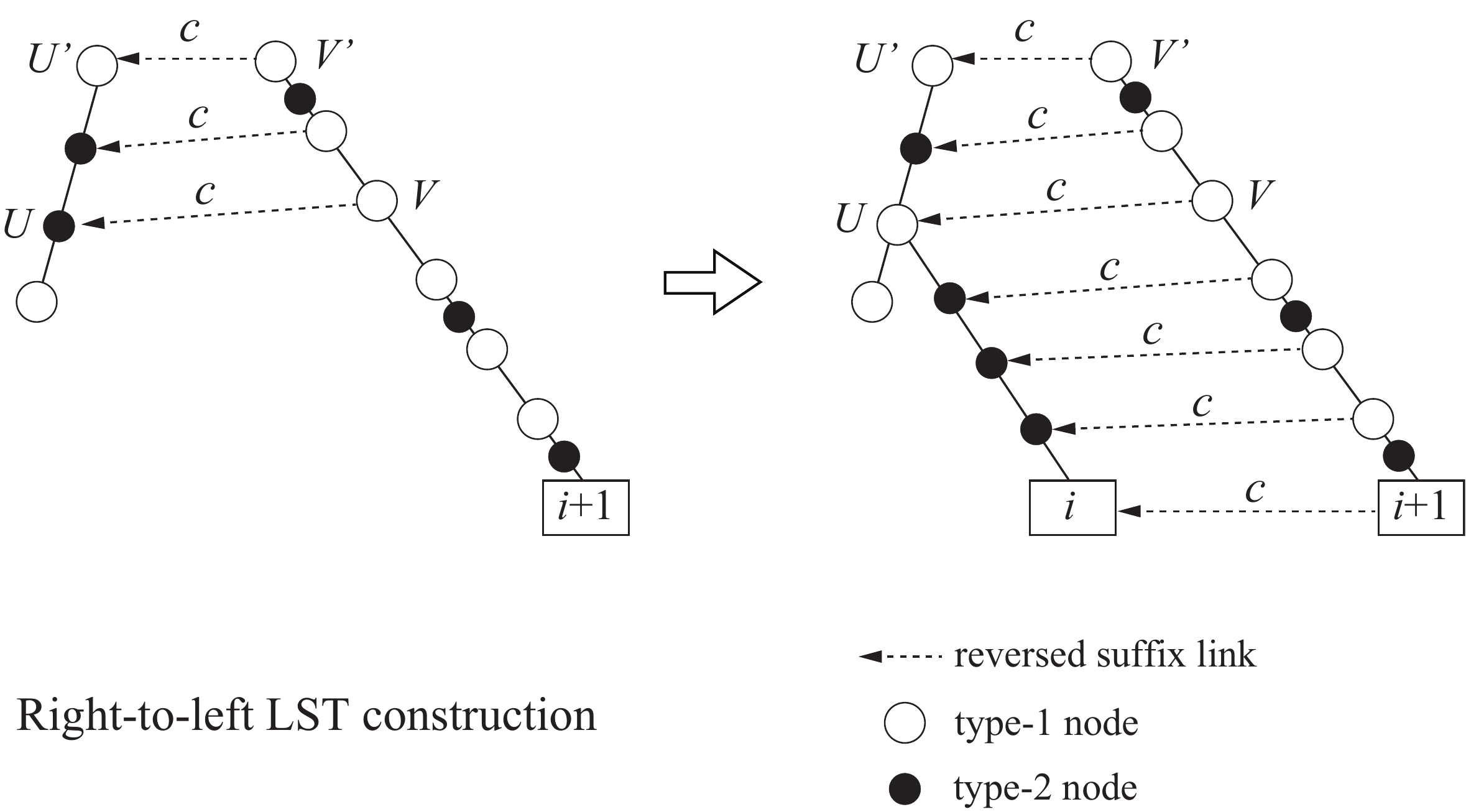}}
	\caption{
		Upper: The DAWG version of Weiner's algorithm when updating the suffix tree for $T[i+1:]$ to the suffix tree for $T[i:]$. Lower: Our right-to-left LST construction when updating $\Tree_{i+1} = \LST(T[i+1:])$ to $\Tree_{i} = \LST(T[i:])$.
	}
	\label{fig:Weiner_LST}
\end{figure}

Let us first recall Weiner's suffix tree contraction algorithm
on which our right-to-left LST construction algorithm is based.
Weiner's algorithm uses the reversed suffix links of the suffix tree
called \emph{hard Weiner links}.
We in particular consider the version of Weiner's algorithm
that also explicitly maintains \emph{soft-Weiner links}~\cite{BreslauerI13} of the suffix tree.
In the suffix tree of a text $T$,
there is a soft-Weiner link for a node $V$ with a symbol $c$
iff $cV$ is a substring of $T$ but $cV$ is not a node in the suffix tree.
It is known that the hard-Weiner links
and the soft-Weiner links are respectively equivalent to
the primary edges and the secondary edges of the
\emph{directed acyclic word graph} (\emph{DAWG})
for the reversal of the input string~\cite{Blumer1985}.

Given the suffix tree for $T[i+1:]$,
Weiner's algorithm walks up from the leaf representing $T[i+1:]$
and first finds the nearest branching ancestor $V$ such that
$aV$ is a substring of $T[i+1:]$,
and then finds the nearest branching ancestor $V'$ such that $cV' = U'$
is also a branching node, where $c = T[i]$.
Then, Weiner's algorithm finds the insertion point for a new leaf for $T[i:]$
by following the reversed suffix link (i.e. the hard-Weiner link) from $V'$ to $U'$,
and then walking down the corresponding out-edge of $U'$ with
the difference of the string depths of $V$ and $V'$.
A new branching node $U$ is made at the insertion point
if necessary.
New soft-Weiner links are created from the nodes between the leaf for $T[i+1:]$
and $V$ to the new leaf for $T[i:]$.

Now we consider our right-to-left LST construction.
See the lower diagram of Figure~\ref{fig:Weiner_LST} for illustration.
The major difference between the DAWG version of Weiner's algorithm
and our LST construction is that
in our LST we explicitly create type-2 nodes which are 
the destinations of the soft-Weiner links.
Hence, in our linear-size suffix trie construction,
for every type-1 node between $V$ and the leaf for $T[i+1:]$,
we explicitly create a unique new type-2 node
on the path from the insertion point to the new leaf for $T[i:]$,
and connect them by the reversed suffix link labeled with $c$.
Also, we can directly access the insertion point $U$
by following the reversed suffix link of $V$,
since $U$ is already a type-2 node before the update.

The above observation also gives rise to the number of type-2 nodes in the LST.
Blumer et al.~\cite{Blumer1985} proved that
the number of secondary edges in the DAWG of any string of length $n$
is at most $n-1$.
Hence we have:
\begin{lemma}
  The number of type-2 nodes in the LST of any string of length $n$
  is at most $n-1$.
\end{lemma}
The original version of Weiner's
suffix tree construction algorithm only maintains a Boolean value
indicating whether there is a soft-Weiner link from each node with each symbol.
We note also that the number of pairs of nodes and symbols for which
the indicators are true is the same as the number of soft-Weiner links
(and hence the DAWG secondary edges).

We have seen that LSTs can be seen as a representation of Weiner's suffix trees
or the DAWGs for the reversed strings.
Another crucial point is that Weiner's algorithm only needs to read
the first symbols of edge labels.
This enables us to easily extend Weiner's suffix tree algorithm
to our right-to-left LST construction.
Below, we will give more detailed properties of LSTs and
our right-to-left construction algorithm.

Let us first observe relations between $\Tree_{i}$ and $\Tree_{i+1}$.
\begin{lemma}\label{lem:rltype2to1}
  Any non-leaf type-1 node $U$ in $\Tree_{i}$ exists in $\Tree_{i+1}$ as a type-1 or type-2 node.
\end{lemma}
\begin{proof}
  If there exist two distinct symbols $a, b \in \Sigma$
  such that $Ua, Ub$ are substrings of $T[i+1:]$,
  then clearly $U$ is a type-1 node in $\Tree_{i+1}$.
  Otherwise, then let $b$ be a unique symbol
  such that $Ub$ is a substring of $T[i+1:]$.
  This symbol $b$ exists since $U$ is not a leaf in $\Tree_{i}$.
  Also, since $U$ is a type-1 node in $\Tree_{i}$,
  there is a symbol $a \neq b$ such that
  $Ua$ is a substring of $T[i:]$.
  Note that in this case $Ua$ is a prefix of $T[i:]$
  and this is the unique occurrence of $Ua$ in $T[i:]$.
  Now, let $U' = U[2:]$.
  Then, $U'a$ is a prefix of $T[i+1:]$.
  Since $U'b$ is a substring of $T[i+1:]$,
  $U'$ is a type-1 node in $\Tree_{i+1}$
  and hence $U$ is a type-2 node in $\Tree_{i+1}$.
\end{proof}
As was described above,
only a single leaf is added to the tree when updating $\Tree_{i+1}$ to $\Tree_{i}$.
The type-2 node of $\Tree_{i}$ that becomes type-1 in $\Tree_{i}$
is the \emph{insertion point} of this new leaf.

\begin{lemma}\label{lem:rlnewbranch}
	Let $U$ be the longest prefix of $T[i:]$
	such that $U$ is a prefix of $T[j:]$ for some $j > i$.
	$U$ is a node in $\Tree_{i+1}$. 
\end{lemma}
\begin{proof}
	If $U=\varepsilon$ then $U$ is the root.
	Otherwise, since $U$ occurs twice or more in $T[i:]$ and $T[i:i+|U|] \ne T[j:j+|U|]$, 
	$U$ is a type-1 node in $\Tree_{i}$.
	 By Lemma~\ref{lem:rltype2to1}, $U$ is a node in $\Tree_{i+1}$.
\end{proof}

By \Cref{lem:rlnewbranch},
we can construct $\Tree_{i}$ by adding a branch on node $U$,
where $U$ is the longest prefix of $T[i:]$ such that $U$ is a prefix of $T[j:]$ for some $j > i$.
This node $U$ is the insertion point for $\Tree_{i}$.
The insertion point $U$ can be found by following the reversed suffix link
labeled by $c$ from the node $U[2:]$ i.e. $U = \Rlink(U[2:],c)$.
Since $U$ is the longest prefix of $T[i:]$ where $U[2:]$ occurs at least twice in $T[i+1:]$,
$U[2:]$ is the deepest ancestor of the leaf $T[i+1:]$ that has the reversed suffix link labeled by $c$.
Therefore, we can find $U$ by checking the reversed suffix links of
the ancestors of $T[i+1:]$ walking up from the leaf.
We call this leaf representing $T[i+1:]$ as the \emph{last leaf} of $\Tree_{i+1}$.

After we find the insertion point, 
we add some new nodes.
First, we consider the addition of new type-1 nodes.
\begin{lemma}\label{lem:rlnewtype1}
	There is at most one type-1 node $U$ in $\Tree_{i}$ such that 
	$U$ is a type-2 node in $\Tree_{i+1}$.
	If such a node $U$ exists, then $U$ is the insertion point of $\Tree_{i}$.
\end{lemma}
\begin{proof}
	Assume there is a type-1 node $U$ in $\Tree_{i}$ such that $U$ is a type-2 node in $\Tree_{i+1}$.
	There are suffixes $UV$ and $UW$ such that $|V| > |W|$ and $V[1] \ne W[1]$.
	Since $U$ is a type-2 node in $\Tree_{i+1}$, $UV = T[i:]$ and $UW = T[j:]$ for some $j > i$.
	Clearly, such a node is the only one which is the branching node.
\end{proof}
From Lemma~\ref{lem:rlnewtype1},
we know that new type-1 node is added at the insertion point
if it is a type-2 node.
The only other new type-1 node is the new leaf representing $T[i:]$.

\begin{figure}[th]
	\centering
	\begin{minipage}[t]{0.49\hsize}
		\centering
		\includegraphics[scale=0.8]{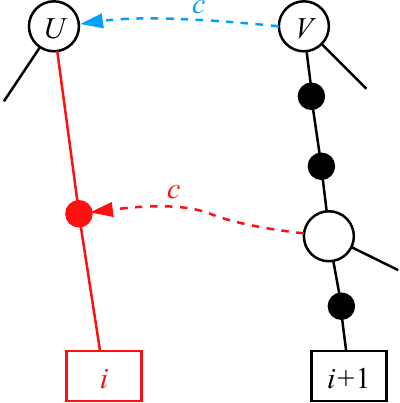}\\
		\ \ \ \small{(a)}
	\end{minipage}
	\begin{minipage}[t]{0.49\hsize}
		\centering
		\includegraphics[scale=0.8]{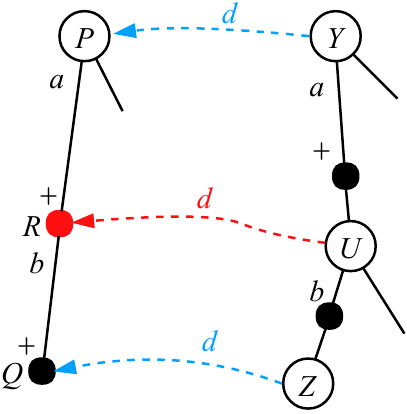}\\
		\ \ \ \small{(b)}
	\end{minipage}
	\caption{
		Illustration of (a) new branch addition and (b) type-2 nodes addition.
		The new nodes, edges, and reverse suffix link are colored red.
	}
	\label{fig:rtol_update}
\end{figure}

Next, we consider the addition of the new branch from the insertion point.
By \Cref{lem:rlnewtype1}, there are no type-1 nodes between the insertion point and
the leaf for $T[i:]$ in $\Tree_{i}$.
Thus, any node $V$ in the new branch is a type-2 node
and this node is added if $V[2:]$ is a type-1 node.
This can be checked by ascending from leaf $T[i+1:]$ to $U[2:]$,
where $U$ is the insertion point.
Regarding the labels of the new branch,
for any new node $V$ and its parent $W$,
the label of $(W,V)$ edge is the same as the label of the first edge between $W[2:]$ and $V[2:]$.
The node $V$ is a $\Plus$-node if $V[2:]$ is a $\Plus$-node or there is a node between $W[2:]$ and $V[2:]$.
\Cref{fig:rtol_update} (a) shows an illustration of the branch addition:
$V$ can be found by traversing the ancestors of $i+1$ leaf.
After we find the insertion point $U = \Rlink(V,c)$,
we add a new leaf $i$ and type-2 nodes for each type-1 node between $i+1$ leaf and $V$.

Last, consider the addition of type-2 nodes when updating the insertion point $U$ to a type-1 node.
In this case, we add a type-2 node $dU$ for any $d \in \Sigma$ such that $dU$ occurs in $T[i:]$.
\begin{lemma}\label{lem:rlnewtype2}
 Let $U$ be the insertion point of $\Tree_{i}$.
 Consider the case where $U$ is a type-2 node in $\Tree_{i+1}$.
 Let $Z$ be the nearest type-1 descendant of $U$
 and $Y$ be the nearest type-1 ancestor of $U$ in $\Tree_{i+1}$.
 For any node $Q$ such that $Q = \Rlink(Z,d)$ for some $d \in \Sigma$,
 $P = \Rlink(Y,d)$ is the parent of $Q$ in $\Tree_{i+1}$
 and there is a type-2 node $R$ between $P$ and $Q$ in $\Tree_{i}$. 
\end{lemma}
\begin{proof}
	First, we prove that $P$ is the parent of $Q$ in $\Tree_{i+1}$.
	Assume on the contrary that $P$ is not the parent of $Q$.
	Then, there is a node $Q[:j] = dZ[:j-1]$ for some $|P| < j < |Q|$.
	Thus, $Z[:j-1]$ is a type-1 ancestor of $Z$ and a type-1 descendant of $Y$,
	however this contradicts the definition of $Z$ or $Y$.
	
	Second, we prove that there is a type-2 node between $P$ and $Q$ in $\Tree_{i}$.
	Since $U$ is a type-2 node in $\Tree_{i+1}$ and $Q = dZ$ is a node in $\Tree_{i+1}$,
	$dU$ occurs in $T[i+1:]$ but is not a node in $\Tree_{i+1}$.
	Since $U$ is a type-1 node in $\Tree_{i}$,
	$dU$ is a type-2 node $\Tree_{i}$.
\end{proof}

See \Cref{fig:rtol_update} (b) for an illustration of type-2 nodes addition.
It follows from Lemma~\ref{lem:rlnewtype2} that we can find the position of new type-2 nodes
by first following the reversed suffix link of the nearest type-1 descendant
$Z$ of $U$ in $\Tree_{i+1}$.
Then, we obtain the parent $P$ of $Z$,
and obtain $Y$ by following the suffix link of $P$.
The string depth of a new type-2 node $R$ equal to the string depth of $U$ plus one.
We can determine whether $R$ is a $\Plus$-node
using the difference of the string depths of $Y$ and $U$.
By Lemma~\ref{lem:rltype2to1},
the total number of type-2 nodes added this way
for all positions $1 \leq i \leq n$ is bounded
by the number of type-1 and type-2 nodes in $\Tree_{n}$ for the whole text $T$.

Algorithm~\ref{alg:rtol} in Appendix
shows a pseudo-code of our right-to-left linear-size suffix trie construction algorithm.
For each symbol $c = T[i]$ read,
the algorithm finds the deepest node $U$ in the path from the root to the last leaf
for $T[i+1:]$ for which $\Rlink(U,c)$ is defined,
by walking up from the last leaf (line \ref{line:rtol_findbranching}).
If the insertion point $\nextNode = \Rlink(U,c)$ is a type-1 node,
the algorithm creates a new branch.
Otherwise (if $\nextNode$ is a type-2 node),
then the algorithm updates $\nextNode$ to type-1 and adds a new branch.
The branch addition is done in lines~\ref{line:rtol_addbranch_start}--\ref{line:rtol_addbranch_end}.

Also, the algorithm adds nodes $R$ such that $R = \Rlink(\nextNode,d)$ for some $d \in \Sigma$ in $\Tree_{i}$.
The algorithm finds the locations of these nodes by checking the reversed suffix links of the nearest type-1 ancestor and descendant of $\nextNode$ by using $\CreateTypeTwo(\nextNode)$.
Let $Y$ be the nearest type-1 ancestor of $\nextNode$ and $Z$ be the nearest type-1 descendant of $\nextNode$.
For a symbol $d$ such that $\Rlink(Z,d)$ is defined,
let $P = \Rlink(Y,d)$ and $Q = \Rlink(Z,d)$:
the algorithm creates type-2 node $R$ and connects it to $P$ and $Q$.

A snapshot of right-to-left LST construction is shown in
Figure~\ref{fig:rtol_example} of Appendix.

We discuss the time complexity of our right-to-left online LST construction algorithm.
Basically, the analysis follows the amortization argument
for Weiner's suffix tree construction algorithm.
First, consider the cost for finding the insertion point for each $i$.
\begin{lemma}\label{lem:findbranchtime}
	Our algorithm finds the insertion point of $\Tree_{i}$ in $O(\log \sigma)$ amortized time.
\end{lemma}
\begin{proof}
	For each iteration, the number of type-1 and type-2 nodes we visit
        from the last leaf to find the insertion point
        is at most $\Depth(L_{i+1})-\Depth(U_i)+1$,
        where $L_{i+1}$ is the leaf representing $T[i+1:]$
        and $U_i$ is the insertion point for the new leaf representing $T[i:]$
        in $\Tree_{i}$, respectively,
        and $\Depth(X)$ denotes the depth of any node $X$ in $\Tree_{i}$.
        See also the lower diagram of Figure~\ref{fig:Weiner_LST} for illustration.
	Therefore, the total number of nodes visited
        is $\sum_{1\le i < n} \Depth(L_{i+1})-\Depth(U_i)+1 \le 2n$.
	Since finding each reversed suffix link takes $O(\log \sigma)$ time,
        the total cost for finding the insertion points for all $1 \leq i \leq n$
        is $O(n \log \sigma)$, which is amortized to $O(\log \sigma)$ per iteration.
\end{proof}

Last, the computation time of a new branch addition in each iteration is as follows.
\begin{lemma}\label{lem:addbranchtime}
  Our algorithm adds a new leaf and new type-2 nodes
  between the insertion point and the new leaf in
  $\Tree_{i}$ in $O(\log \sigma)$ amortized time.
\end{lemma}
\begin{proof}
  Given the insertion point for $\Tree_{i}$,
  it is clear that we can insert a new leaf in $O(\log \sigma)$ time.
  For each new type-2 node in the path from the insertion point and the new leaf for $T[i:]$,
  there is a corresponding type-1 node in the path above the last leaf $T[i+1:]$
  (see also the lower diagram of Figure~\ref{fig:Weiner_LST}).
  Thus the cost for inserting all type-2 nodes can be charged to
  the cost for finding the insertion point for $\Tree_{i}$,
  which is amortized $O(\log \sigma)$ per a new type-2 node by Lemma~\ref{lem:findbranchtime}.  
\end{proof}

By Lemmas~\ref{lem:findbranchtime} and~\ref{lem:addbranchtime},
we get the following theorem:
\begin{theorem}\label{theorem:rtollstconstructiontime}
	Given a string $T$ of length $n$, our algorithm constructs $\LST(T)$ in $O(n \log \sigma)$ time and $O(n)$ space online, by reading $T$ from the right to the left.
\end{theorem}

\section{Left-to-right online algorithm}

In this section, we present an algorithm that constructs the linear-size suffix trie of a text $T$
by reading the symbols of $T$ from the left to the right.
Our algorithm constructs a slightly-modified
data structure called the pre-LST defined as follows:
	The pre-LST $\PLST(T)$ of a string $T$ is a subgraph of $\STrie(T)$ consisting of two types of nodes,
	\begin{enumerate}
		\item Type-1: The root, branching nodes, and leaves of $\STrie(T)$.
		\item Type-2: The nodes of $\STrie(T)$ that are not type-1 nodes and their suffix links point to type-1 nodes.
	\end{enumerate}
The main difference between $\PLST(T)$ and $\LST(T)$ is the definition of type-1 nodes.
While $\LST(T)$ may contain non-branching type-1 nodes that correspond to 
non-branching internal nodes of $\STree(T)$ which represent repeating suffixes,
$\PLST(T)$ does not contain such type-1 nodes.
When $T$ ends with a unique terminal symbol $\$$,
the pre-LST and LST of $T$ coincide.

Our algorithm is based on Ukkonen's suffix tree construction algorithm~\cite{Ukkonen1995}.
For each prefix $T[:i]$ of $T$,
there is a unique position $k_i$ in $T[:i]$
such that $T[k_i:i]$ occurs twice or more in $T[:i-1]$
but $T[k_i-1:i]$ occurs exactly once in $T[:i]$.
In other words, $T[k_i-1:i]$ is the shortest suffix of $T[:i]$ that is represented
as a leaf in the current pre-LST $\PLST(T[:i])$,
and $T[k_i:i]$ is the longest suffix of $T[:i]$ that is
represented in the ``inside'' of $\PLST(T[:i])$.
The location of $\PLST(T[:i])$ representing the longest repeating suffix $T[k_i:i]$
of $T[:i]$ is called the \emph{active point}, as in the
Ukkonen's suffix tree construction algorithm.
We also call $k_i$ the \emph{active position} for $T[:i]$.
Our algorithm keeps track of the location for the active point (and the active position)
each time a new symbol $T[i]$ is read for increasing $i = 1, \ldots, n$.
We will show later that the active point can be
maintained in $O(\log \sigma)$ amortized time per iteration,
using a similar technique to our pattern matching algorithm on LSTs
in Lemma~\ref{lem:patmatch}.
In order to ``neglect'' extending the leaves that already exist in the current tree,
Ukkonen's suffix tree construction algorithm uses
the idea of \emph{open leaves} that do not explicitly maintain
the lengths of incoming edge labels of the leaves.
However, we cannot adapt this open leaves technique to construct pre-LST directly,
since we need to add type-2 node on the incoming edges of some leaves.
Fortunately, there is a nice property on the pre-LST so we can update it efficiently.
We will discuss the detail of this property later.
Below, we will give more detailed properties of pre-LSTs and
our left-to-right construction algorithm.

Let $\PTree_{i} = \PLST(T[:i])$ be the pre-LST of $T[:i]$.
Our algorithm constructs $\PTree_{i}$ from $\PTree_{i-1}$ incrementally when a new symbol $c = T[i]$ is read.

There are two kinds of leaves in $\PLST(T[:i])$,
the one that are $\Plus$-nodes and the other ones that are not $\Plus$-nodes.
There is a boundary in the suffix link chain of the leaves
that divides the leaves into the two groups, as follows:
\begin{lemma}\label{lem:leafpoint}
	Let $T[j:i]$ be a leaf of $\PTree_{i}$, for $1 \le j < k$.
	There is a position $l$ such that $T[j:i]$ is a $\Plus$-node for $1 \le j < l$
	and not a $\Plus$-node for $l \le j < k_i$.
\end{lemma}
\begin{proof}
	Assume on the contrary there is a position $j$ such that $T[j:i]$ is not a $\Plus$-node
	and $T[j+1:i]$ is a $\Plus$ node.
	Since $T[j:i]$ is not a $\Plus$-node, $T[j:i-1]$ is a node.
	By definition, $T[j+1:i-1]$ is also a node.
	Thus $T[j+1:i]$ is not a $\Plus$-node, which is a contradiction.
\end{proof}
Intuitively, the leaves that are $\Plus$-nodes in $\PTree_{i}$
are the ones
that were created in the last step of the algorithm with the last read symbol $T[i]$.

\begin{figure}[t]
	\centering
	\includegraphics[scale=0.8]{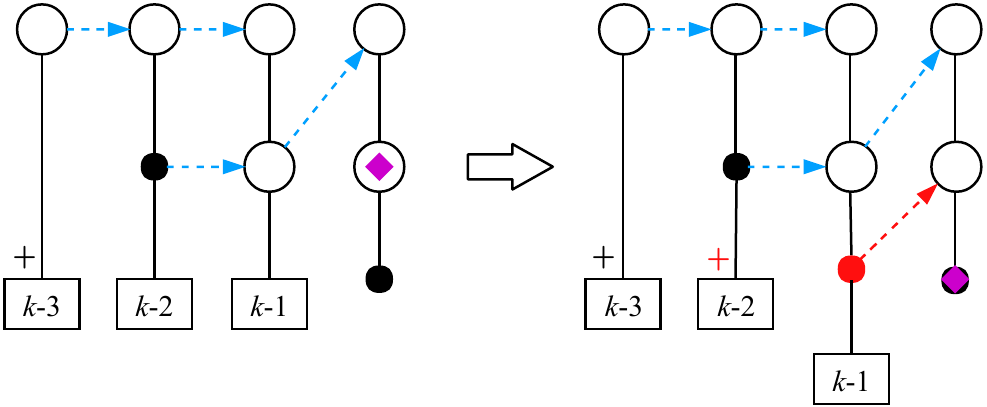}
	\caption{
		Illustration for updating the parts of $\PTree_{i-1}$ that correspond to $T[j:i-1]$ for $j < k_i$.
		The purple diamond shows the active point.
		The new $\Plus$ sign, node, and its suffix link are colored red.
	}
	\label{fig:ltor_updateleft}
\end{figure}

When updating $\PTree_{i-1}$ into $\PTree_{i}$,
the active position $k_{i-1}$ for $T[:i-1]$ divides the suffixes $T[j:i-1]$ into two parts,
the $j < k_{i-1}$ part and the $j \ge k_{i-1}$ part.
First, we consider updating the parts of $\PTree_{i-1}$ that correspond to $T[j:i-1]$ for $j < k_{i-1}$.
\begin{lemma}\label{lem:lropenleaf}
	For any leaf $T[j:i-1]$ of $\PTree_{i-1}$ with $j < k_{i-1} - 1$,
        $T[j:i-1]$ is implicit in $\PTree_{i}$.
\end{lemma}
\begin{proof}
	Consider updating $\PTree_{i-1}$ to $\PTree_{i}$.
	$T[k_{i-1}-1:i-1]$ cannot be a type-1 node in $\PTree_{i}$.
	Therefore, $T[k_{i-1}-2:i-1]$ is implicit in $\PTree_{i}$.
	$T[j:i-1]$ for $j < k_{i-1}-1$ are also implicit.
\end{proof}

\begin{lemma}\label{lem:lrplusleaf}
  If $T[j:i-1]$ is a leaf in $\PTree_{i-1}$,
  then $T[j:i]$ is a $\Plus$-leaf in $\PTree_{i}$, where $1 \le j < k_{i-1}-1$.
\end{lemma}
\begin{proof}
  Assume on the contrary that $T[j:i-1]$ is a leaf in $\PTree_{i-1}$
  but $T[j:i]$ is not a $\Plus$-leaf in $\PTree_{i}$.
  Then $T[j:i-1]$ is a node in $\PTree_{i}$.
  Since $T[j:i-1]$ is a leaf in $\PTree_{i-1}$,
  $T[j:i-1]$ cannot be a type-1 node in  $\PTree_{i}$.
  Moreover, $T[j+1:i-1]$ is a leaf in $\PTree_{i-1}$,
  thus $T[j+1:i-1]$ cannot be a type-1 node in $\PTree_{i}$ and $T[j:i-1]$ cannot be a type-2 node in $\PTree_{i}$.
  Therefore, $T[j:i-1]$ is neither type-1 nor type-2 node in $\PTree_{i}$, which contradicts the assumption.
\end{proof}
\Cref{lem:lropenleaf} shows that we do not need to add nodes on the leaves of $\mathcal{P}_{i-1}$ besides $T[k-1:i]$ leaf
and \Cref{lem:lrplusleaf} shows that we can update all leaves $T[j:i]$ for $l \le j < k-1$ to a $\Plus$-leaf.
Therefore, besides the leaf for $T[k-1:i]$, once we update a leaf to $\Plus$ node, 
we do not need to update it again.
\Cref{fig:ltor_updateleft} shows an illustration of how to update this part.

Next, we consider updating the parts of $\PTree_{i-1}$ that correspond to $T[j:i-1]$ for $j \ge k_{i-1}$.
If $T[k_{i-1}:i]$ exists in the current LST (namely $T[k_{i-1}:i]$ occurs in $T[:i-1]$),
then the $j \ge k_{i-1}$ part of the current LST does not need to be updated.
Then we have $k_{i} = k_{i-1}$ and $T[k_{i}:i]$ is the active point of $\PTree_{i}$.
Otherwise, we need to create new nodes recursively from the active point
that will be the parents of new leaves.
There are three cases for the active point $T[k_{i-1}:i-1]$ in $\PTree_{i-1}$:

\textbf{Case 1}: $T[k_{i-1}:i-1]$ is a type-1 node in $\PTree_{i-1}$.
Let $T[p:i]$ be the longest suffix of $T[k_{i-1}:i]$ that exists in $\PTree_{i-1}$.
Since $T[k_{i-1}:i-1]$ is a type-1 node, $T[j:i-1]$ is also a type-1 node for $k_{i-1} \le j < p$.
Therefore, we can obtain $\PTree_{i}$ by adding a leaf from the node representing
$T[j:i-1]$ for every $k \le j < p$,
with edge label $c$ by following the suffix link chain from $T[k_{i-1}:i-1]$.
In this case, we only need to add one new type-2 node, which is $T[k_{i-1}-1:i-1]$
that is connected to the type-1 node $T[k_{i-1}:i-1]$ by the suffix link.
Moreover, $p$ will be the active position for $T[:i]$, namely $k_{i} = p$.

\textbf{Case 2}: $T[k_{i-1}:i-1]$ is a type-2 node in $\PTree_{i-1}$.
Similarly to Case 1,
we add a leaf from the node representing $T[j:i-1]$ for every $k_{i-1} \le j < p$
with edge label $c$ by following the suffix link chain from $T[k_{i-1}:i-1]$,
where $p$ is defined as in Case 1..
Then, $T[k_{i-1}:i-1]$ becomes a type-1 node,
and a new type-2 node $T[k_{i-1}-1:i-1]$ is added and is connected to
this type-1 node $T[k_{i-1}:i-1]$ by the suffix link.
Moreover, for any symbol $d$ such that $dT[k_{i-1}:i-1]$ is a substring
of $T[:i]$, a new type-2 node for $dT[k_{i-1}:i-1]$ is added to the tree,
and is connected by the suffix link to this new type-1 node $T[k_{i-1}:i-1]$.
These new type-2 nodes can be found in the same way as
in Lemma~\ref{lem:rlnewtype2} for our right-to-left LST construction.
Finally, $p$ will become the active position for $T[:i]$, namely $k_i = p$.

\textbf{Case 3}: $T[k_{i-1}:i-1]$ is implicit in $\PTree_{i-1}$.
In this case, there is a position $p > k_{i-1}$ such that $T[p:i-1]$ is a type-2 node.
We create new type-1 nodes $T[j:i-1]$ and leaves $T[j:i]$ for $k \le j < p$,
then do the same procedure as Case 2 for $T[j:i-1]$ for $p \le j$.

\Cref{fig:ltor_updateright} shows an illustration of how to add new leaves.
Algorithm~\ref{alg:ltor} shows a pseudo-code of our 
left-to-right online algorithm for constructing LSTs.
In Case 1 or Case 2, the algorithm
checks whether there is an out-going edge labeled with $c = T[i]$,
and performs the above procedures 
(lines~\ref{line:ltor_addbranch_start}--\ref{line:ltor_addbranch_end}).
In Case 3, we perform $\ReadEdge$ to check if the active point can
proceed with $c$ on the edge.
The function $\ReadEdge$ returns the location of the new active point and sets $\Flag=\False$ if there is no mismatch, or it returns the mismatching position and sets $\Flag=\True$
if there is a mismatch.
If there is no mismatch, then we just update the $T[j:i-1]$ part of the current LST
for $j < k_{i-1}$.
Otherwise, then we create new nodes as explained in Case 3,
by $\Split$ in the pseudo-code.

\begin{figure}[t]
	\begin{minipage}[c]{0.54\hsize}
	\centering
	\includegraphics[scale=0.72]{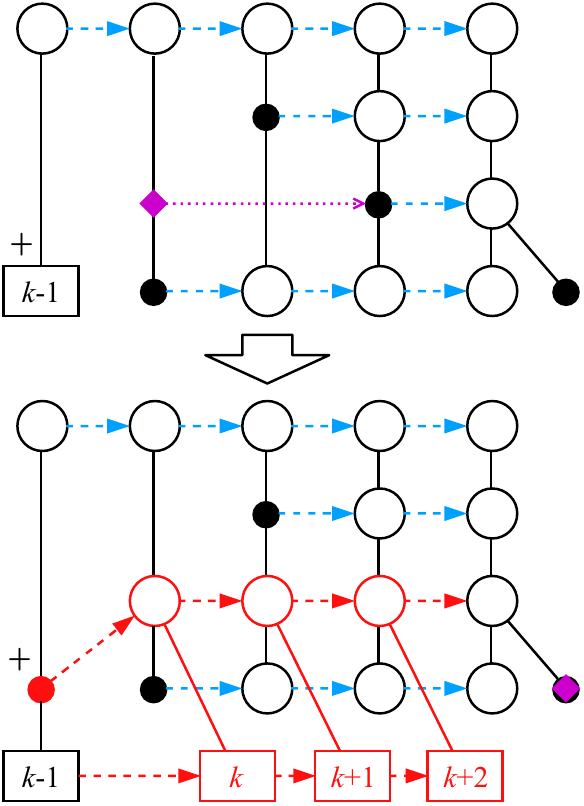}
	\caption{
		Illustration for updating the parts of $\PTree_{i-1}$ that correspond to $T[j:i-1]$ for $j \ge k_{i-1}$.
		The purple diamond and arrow show the active point and its virtual position when reading the edge.
		The new branches, nodes, and their suffix links are colored red.
	}
	\label{fig:ltor_updateright}
	\end{minipage}
	\hfill
	\begin{minipage}[c]{0.41\hsize}
	\centering
	\includegraphics[scale=0.41]{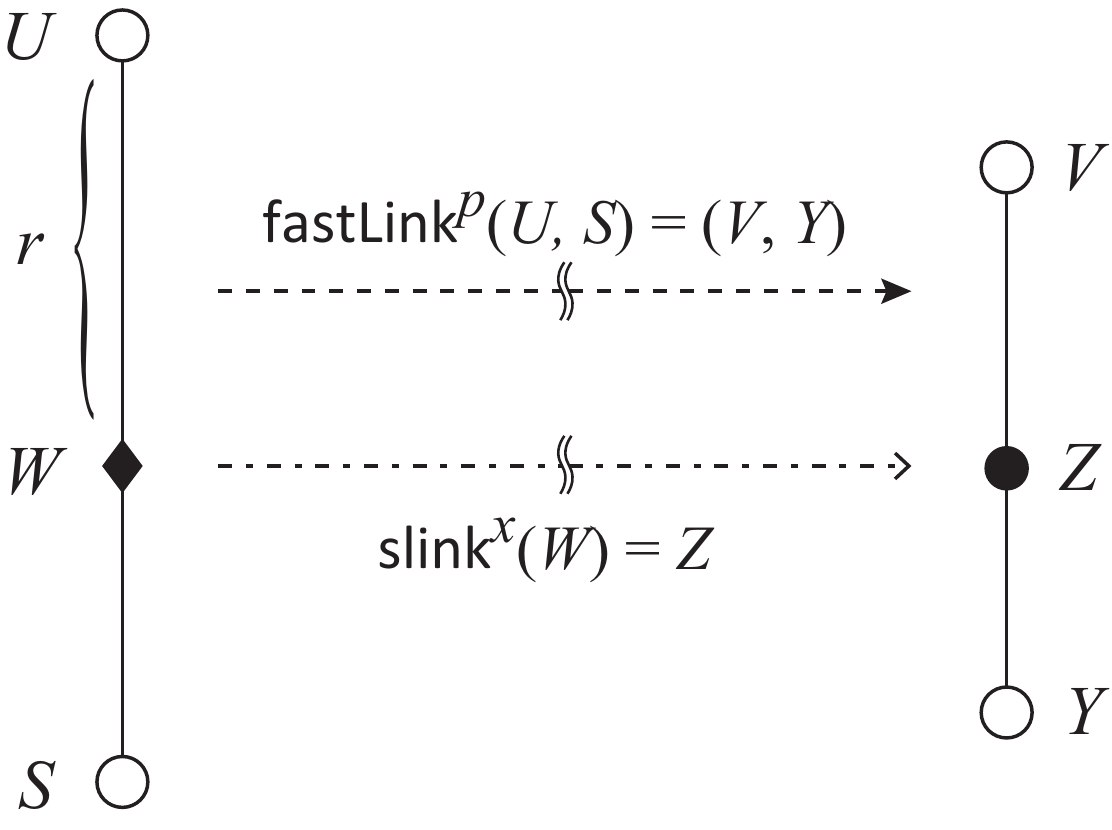}
	\caption{Illustration for our analysis of the cost to maintain the active point. The diamond shows the current location of the active point. New leaves will be created from $W$ to $Z$ by following the (virtual) suffix link chain of length $x$. When we have reached the edge $(V, Y)$, we have already retrieved the corresponding prefix of the label between $U$ and $W$. The rest of the label can be retrieved by at most $r$ applications of $\FLink$ from edge $(V, Z)$.}
\label{fig:active_point}
\end{minipage}
\end{figure}

A snapshot of right-to-left LST construction is shown in
Figure~\ref{fig:ltor_example} of Appendix.

We discuss the time complexity of our left-to-right online construction for LSTs.
To maintain the active point for each $T[:i]$,
we use a similar technique to Lemma~\ref{lem:patmatch}.
\begin{lemma} \label{lem:active_point_maintain}
The active point can be maintained in $O(f(n)+\log \sigma)$ amortized time
per each iteration,
where $f(n)$ denotes the time for accessing $\FLink$ in our growing LST.
\end{lemma}

\begin{proof}
We consider the most involved case where 
the active point lies on an implicit node $W$
on some edge $(U, S)$ in the current LST.
The other cases are easier to show.
Let $r = |W|-|U|$, i.e., the active point is hanging off $U$ with string depth $r$.
Let $Z$ be the type-2 node from which a new leaf will be created.
By the monotonicity on the suffix link chain
there always exists such a type-2 node.
See Figure~\ref{fig:active_point} for illustration.
Let $p$ be the number of applications of $\FLink$ from edge $(U, S)$
until reaching the edge $(V, Y)$ on which $Z$ lies.
Since such a type-2 node $Z$ always exists,
we can sequentially retrieve the first $r$ symbols with at most $r$ applications of $\FLink$
by the same argument to Lemma~\ref{lem:patmatch}.
Thus the number of applications of $\FLink$
until finding the next location of the active point is bounded by $p+r$.
If $x$ is the number of (virtual) suffix links from $W$ to $Z$,
then $p \leq x$ holds.
Recall that we create at least $x+1$ new leaves by following the (virtual) suffix link chain
from $W$ to $Z$.
Now $r$ is charged to the number of text symbols read on the edge from $U$,
and $p$ is charged to the number of newly created leaves,
and both of them are amortized constant as in Ukkonen's suffix tree algorithm.
Thus the number of applications of $\FLink$ is amortized constant,
which implies that it takes $O(f(n)+\log \sigma)$ amortized time to maintain the active point.
\end{proof}

To maintain $\FLink$ in our growing (suffix link) tree,
we use the nearest marked ancestor (NMA) data structure~\cite{Alstrup1998}
that allows marking, unmarking, and NMA query in an online manner in $O(\log n / \log \log n)$ time each, using $O(n)$ space on a dynamic tree of size $n$.
By maintaining the tree of suffix links of edges
enhanced with the NMA data structure, 
we have $f(n) = O(\log n / \log \log n)$ for Lemma~\ref{lem:active_point_maintain}.
This leads to the final result of this section.

\begin{theorem}\label{theorem:ltorlstconstructiontime}
	Given a string $T$ of length $n$, our algorithm constructs $\LST(T)$ in $O(n (\log \sigma + \log n / \log \log n))$ time and $O(n) $space online, by reading $T$ from the left to the right.
\end{theorem}

\section{Conclusions and Future Work}

In this paper we proposed a right-to-left online algorithm which constructs linear-size suffix trees (LSTs) in $O(n \log \sigma)$ time and $O(n)$ space,
and a left-to-right online algorithm which constructs LSTs in $O(n (\log \sigma + \log n / \log \log n))$ time and $O(n)$ space,
for an input string of length $n$ over an ordered alphabet of size $\sigma$.
Unlike the previous construction algorithm by Crochemore et al.~\cite{Crochemore2016},
our algorithms do not construct suffix trees as an intermediate structure,
and do not require to store the input string.
Fischer and Gawrychowski~\cite{0001G15} showed
how to build suffix trees in a right-to-left online manner
in $O(n(\log\log n + \log^2 \log \sigma / \log \log \log \sigma))$ time
for an integer alphabet of size $\sigma = n^{O(1)}$.
It might be possible to extend their result to 
our right-to-left online LST construction algorithm.
An improvement of the running time of left-to-right online LST construction
is also left for future work.

Takagi et al.~\cite{Takagi2017} proposed \emph{linear-size CDAWGs} (\emph{LCDAWG}),
which are edge-labeled DAGs obtained by merging isomorphic subtrees of LSTs.
They showed that the LCDAWG of a string $T$ takes only $O(e+e')$ space,
where $e$ and $e'$ are respectively the numbers of right and left extensions of
the maximal repeats in $T$,
which are always smaller than the text length $n$.
Belazzougui and Cunial~\cite{BelazzouguiC17} proposed
a very similar CDAWG-based data structure that uses only $O(e)$ space.
It is not known whether these data structures can be efficiently constructed in an online manner,
and thus it is interesting to see if our algorithms can be extended to these data structures.
The key idea to both of the above CDAWG-based structures is to implement
edge labels by \emph{grammar-compression} or \emph{straight-line programs},
which are enhanced with efficient grammar-compressed data structures~\cite{GasieniecKPS05,BilleLRSSW15}.
In our online setting, the underlying grammar needs to be dynamically updated,
but these data structures are static.
It is worth considering if these data structures can be efficiently dynamized
by using recent techniques such as e.g.~\cite{GawrychowskiKKL18}.

\clearpage

\bibliographystyle{plainurl}
\bibliography{ref}

\newpage
\appendix
\appendix

\section{Supplementary Figures}

\begin{figure}[!ht]
	\centering
	\includegraphics[scale=0.93]{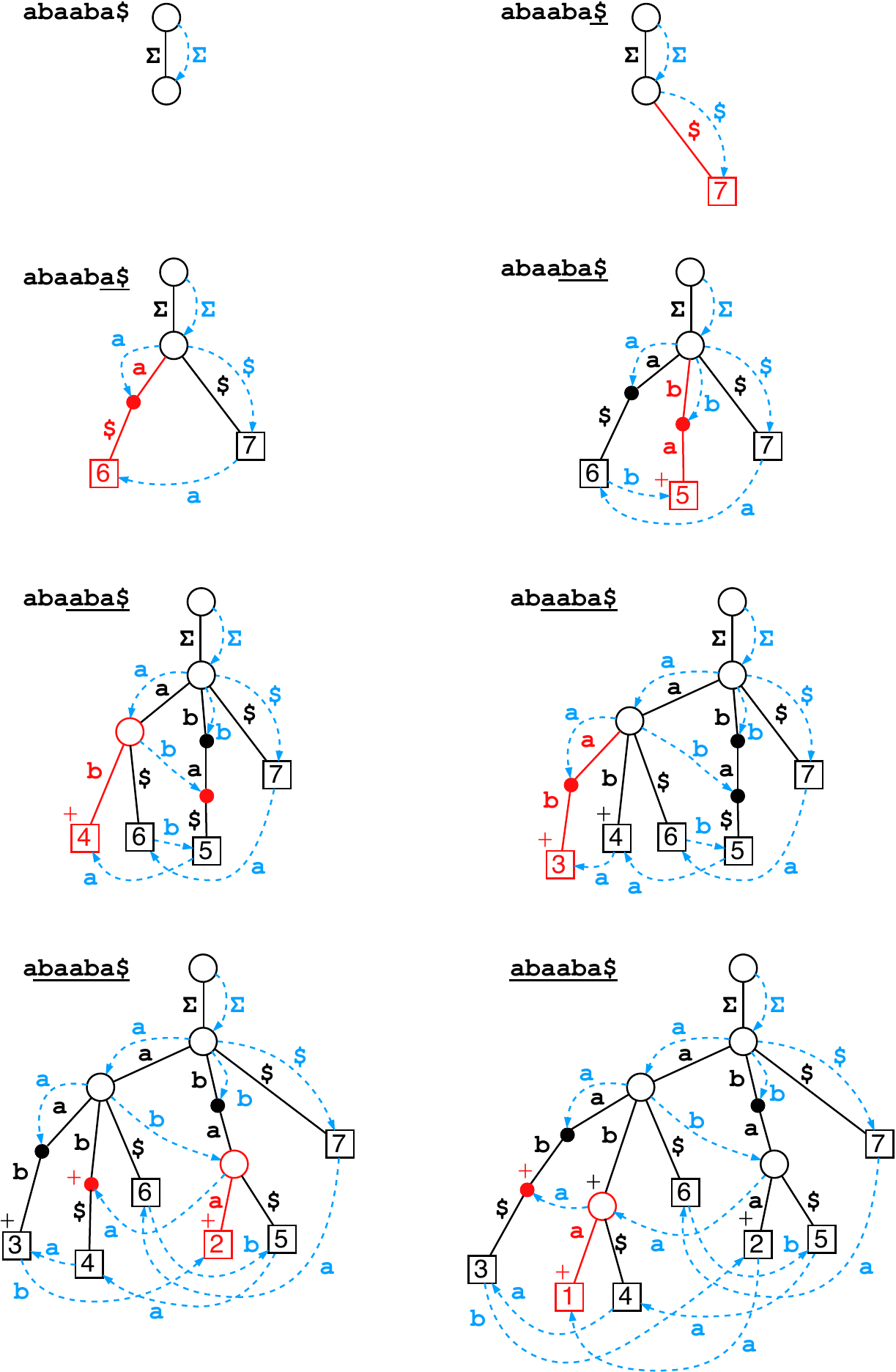}
	\vspace{3mm}
	\caption{
                A snapshot of right-to-left online construction of $\LST(T)$ with $T = \mathtt{abaaba\texttt{\$}}$ by Algorithm~\ref{alg:rtol}.
		The white circles show Type-1 nodes, the black circles show Type-2 nodes,
		and the rectangles show leaves. 
		The reverse suffix links and its label are colored blue.
		The new branches and nodes are colored red.
	}
	
	\label{fig:rtol_example}
\end{figure}

\begin{figure}[!ht]
	\centering
	\includegraphics[scale=0.93]{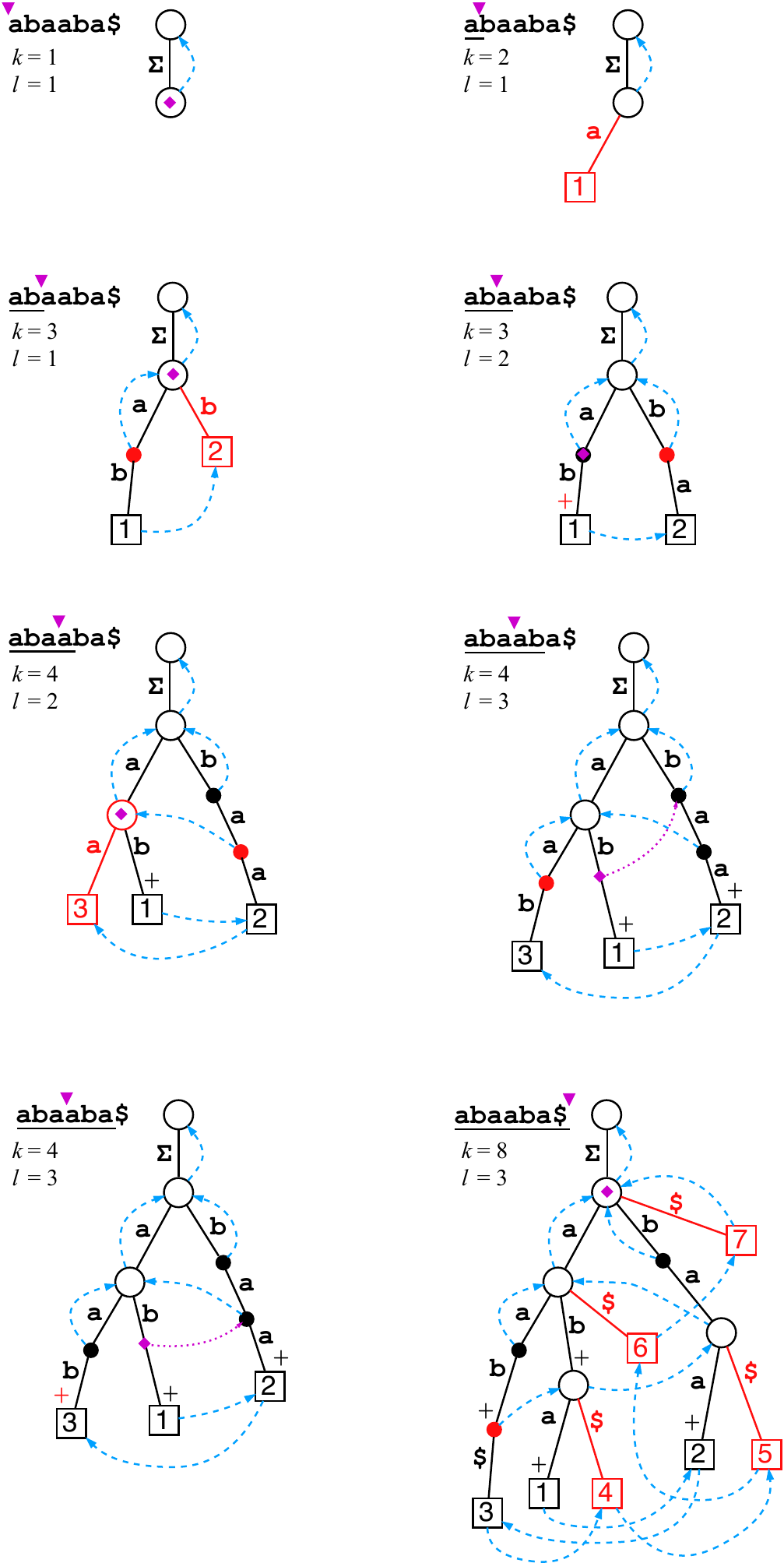}
	\vspace{3mm}
	\caption{
		A snapshot of left-to-right online construction of $\LST(T)$ with $T = \mathtt{abaaba\texttt{\$}}$ by Algorithm~\ref{alg:ltor}.
		The purple diamond and arrow represent the active point and its virtual position when reading the edge label.
		The suffix links are colored blue.
		The new branches and nodes are colored red.
		$k$ is the active position and $l$ is the boundary position
                for $\Plus$-leaves and non-$\Plus$ leaves defined in Lemma~\ref{lem:leafpoint}.
	}
	
	\label{fig:ltor_example}
\end{figure}

\clearpage

\section{Pseudo-codes}


\begin{algorithm2e}[!ht]
	\caption{Right-to-left linear-size suffix trie construction algorithm}
	\label{alg:rtol}
	\SetVlineSkip{0.5mm}
	$\Child(\bot,c) \ot \Root$ for any $c \in \Sigma$;
	$\Rlink(\bot,c) \ot \Root$ for all $c \in \Sigma$\;
	$\PrevInsPoint \ot \bot$;
	$\PrevLeaf \ot root$;
	$\PrevLabel \ot \NULL$\;
	\For{$i=n$ \textbf{to} $1$}{
		$c \ot T[i]$;
		$U \ot \PrevInsPoint$\; 
		\lWhile{$\Rlink(U,c) = \NULL$}{
			$U \ot \Parent(U)$}\label{line:rtol_findbranching}
		$\nextNode \ot \Rlink(U,c)$\;
		\If{$\Type(\nextNode) = 2$}{
			$\CreateTypeTwo(\nextNode)$\;
			$\Type(\nextNode) \ot 1$\; 
		}
		create a leaf $\newLeaf$\; \label{line:rtol_addbranch_start}
		$V \ot \PrevLeaf$;
		$U \ot \PrevInsPoint$;
		$Y \ot \newLeaf$\;
		\While{$\Rlink(U,c) = \NULL$}{
			create a type-2 node $X$\;
			\leIf{$U = \PrevInsPoint$}{
				$a = \PrevLabel$}
			{
				$a = \Label(U,V)$}
			\lIf{$\Plus(V) = \True$ \bf{or} $\Child(U,a) \ne V$}{$\Plus(Y) \ot \True$}
			$\Child(X,a) \ot Y$;
			$\Rlink(U,c) \ot X$;
			$Y \ot X$\;
			$V \ot U$\;
			\lRepeat{$\Type(U) = 1$}{$U \ot \Parent(U)$}
		}
		\leIf{$U = \bot$}{
			$a = c$}
		{
			$a = \Label(U,V)$}
		\lIf{$\Plus(V) = \True$ \bf{or} $\Child(U,a) \ne V$}{$\Plus(Y) \ot \True$}
		$\Child(\nextNode,a) \ot Y$\;\label{line:rtol_addbranch_end}
		$\PrevInsPoint \ot \nextNode$;
		$\PrevLeaf \ot \newLeaf$;
		$\PrevLabel \ot a$\;
	}
\end{algorithm2e}

\begin{algorithm2e}[!ht]
	\caption{$\CreateTypeTwo(U)$}
	\label{alg:createtype2}
	\SetVlineSkip{0.5mm}
	\Fn{$\CreateTypeTwo(U)$}{
		$V \ot U$;
		$b = \Label(U)$;
		$Z \ot \TChild(U,b)$\;
		\For{$d$ \rm{such that} $\Rlink(Z,d) \ne \NULL$}{
			$Q \ot \Rlink(Z,d)$\;
			$P \ot \Parent(Q)$\;
			\If{$\Slink(P) \ne \NULL$}{
				$a \ot \Label(P,Q)$\;
				$Y \ot \Slink(P)$\;
				create a type-2 node $R$\;
				$\Child(P,a) \ot R$;
				$\Child(R,b) \ot Q$\;
				\lIf{$\Child(Y,a) \ne U$ \bf{or} $\Plus(\Child(Y,a)) = \True$}{$\Plus(R) \ot \True$}
				\lIf{$\Child(U,b) \ne Z$ \bf{or} $\Plus(\Child(U,b)) = \True$}{$\Plus(Q) \ot \True$}
			}
		}
	}
	
\end{algorithm2e}

\begin{algorithm2e}[!ht]
	\caption{Left-to-right linear-size suffix trie construction algorithm}
	\label{alg:ltor}
	\SetVlineSkip{0.5mm}
	create $\Root$ and $\bot$;
	$\Child(\bot,c) \ot \Root$ for any $c \in \Sigma$\;
	$\activeNode = \Root$;
	$i \ot 1$;
	$l \ot 1$;
	$k \ot 1$\;
	\While{$i \le n$}{
		$c \ot T[i]$\;
		\If{$\Child(\activeNode,c) \ne \NULL$}{
			$V \ot \Child(\activeNode,c)$\;
			$(U,i',\Flag) \ot \ReadEdge((\activeNode,V),i)$\;
			\If{$\Type(\activeNode) = 1$}{
				create a type-2 node $W$\;
				$V \ot \Parent(\Leaf[k-1])$\;
				$\Child(W,c) \ot \Leaf[k-1]$;
				$\Child(V,\Label(V,\Leaf[k-1])) \ot W$\;
				$\Plus(W,c) \ot \Plus(\Leaf[k-1])$;
				$\Slink(W) \ot \activeNode$;
			}
			\lElse{
				$\Plus(\Leaf[k-1]) \ot \True$}
			\lWhile{$j \ne k-1$}{
				$\Plus(\Leaf[l]) \ot \True$;
				$l \ot l + 1$}
			\If{$\Flag = \False$}{
				\lIf{$\Plus(U) = \True$}{
					$\Plus(\Leaf[k-1]) \ot \True$}
			}
			\lElse{
				$\Split(U,\activeNode,c,i,i')$}
			$\activeNode \ot U$;	
			$i \ot i'$\;
		}
		\Else{\label{line:ltor_addbranch_start}			
			\If{$\Type(\activeNode) = 2$}{
				$\CreateTypeTwo(\activeNode)$;
				$\Type(\activeNode) \ot 1$\;
			}
			\lWhile{$l \ne k-1$}{
				$\Plus(\Leaf[l]) \ot \True$;
				$l \ot l + 1$}
			create a type-2 node $W$;
			$V \ot \Parent(\Leaf[k-1])$\;
			$\Child(W,c) \ot \Leaf[k-1]$;
			$\Child(V,\Label(V,\Leaf[k-1])) \ot W$\;
			$\Plus(W,c) \ot \Plus(\Leaf[k-1])$;
			$\Slink(W) \ot \activeNode$\;
			\While{$\Child(\activeNode,c) = \NULL$}{
				create a leaf $U$\;
				$\Child(\activeNode,c) \ot U$;
				$\Slink(\Leaf[k-1]) \ot U$\;
				$k \ot k + 1$;
				$\Leaf[k-1] \ot U$;
				$\activeNode = \Slink(\activeNode)$\;
			}
		}\label{line:ltor_addbranch_end}
	}
\end{algorithm2e}

\begin{algorithm2e}[!ht]
	\caption{$\ReadEdge((U,V),i)$}
	\label{alg:readedge}
	\SetVlineSkip{0.5mm}
	\Fn{$\ReadEdge(U,V,i)$}{
		\While{$U \ne V$}{
			$c \ot T[i]$\;
			\lIf{$\Child(U,c) = \NULL$}{
				$\Ret$ $(U,i,\True)$}
			\Else{
				\If{$\Plus(\Child(U,c)) = \True$}{
					$(W,i,\Flag) \ot \ReadEdge(\FLink(U,\Child(U,c)),i)$\;
					\lIf{$\Flag = true$}{
						$\Ret$ $(W,i,\True)$}
					$U \ot W$\;
				}
				\lElse{
					$U \ot \Child(U,c)$;
					$i \ot i + 1$}
			}
		}
		$\Ret$ $(U,i,\False)$\;
	}
	
\end{algorithm2e}

\begin{algorithm2e}[!ht]
	\caption{$\Split(U,X,a,i,i')$}
	\label{alg:split}
	\SetVlineSkip{0.5mm}
	\Fn{$\Split(U,X,a,i,i')$}{
		$b = \Label(U,\Child(U))$;
		$c' \ot T[i']$\;
		create a type-1 node $W$\;
		$V \ot \Parent(\Leaf[k-1])$\;
		$\Child(W,c) \ot \Leaf[k-1]$;
		$\Child(V,\Label(V,\Leaf[k-1])) \ot W$\;
		$\Plus(W) \ot \Plus(\Leaf[k-1])$;
		$\LastNode \ot W$\;
		$k \ot k+1$;
		$Y' \ot \Leaf[k-1]$\;
		\While{$X \ne U$}{
			\lIf{$\Type(x) = 1$}{
				$Y \ot \Child(X,a)$}
			$d = \STriedepth(Y) - \STriedepth(X)$\;
			\While{$d < i' - i$}{
				$X \ot Y$;
				$i \ot i + d$\;
				$Y \ot \Child(X)$;	
				$d \ot \STriedepth(Y) - \STriedepth(X)$\;
			}
			\If{$X \ne U$}{
				create a type-2 node $Z$;
				create a leaf $Y'$;
				$a \ot \Label(X,Y)$\;
				$\Child(X, a) \ot Z$;
				$\Child(Z, b) \ot Y$;
				$\CreateTypeTwo(Z)$\;
				$\Type(Z) \ot 1$;
				$\Child(Z, c') \ot Y'$\;
				\lIf{$i' - 1 > 1$}{$\Plus(Z) \ot \True$}
				\lIf{$d - (i' - 1) > 1$}{$\Plus(Y) \ot \True$}
				$\Slink(\LastNode) \ot Z$;
				$\Slink(\Leaf[k-1]) \ot Y'$\;
				$k \ot k + 1$;
				$\Leaf[k-1] \ot Y'$\;
				$\LastNode \ot Z$;
				$X \ot \Slink(X)$\;
			}
		}
		$\Slink(\LastNode) \ot U$\;
	}
\end{algorithm2e}

\begin{algorithm2e}[!ht]
	\caption{Fast pattern matching algorithm with the LST}
	\label{alg:patmatch}
	\SetVlineSkip{0.5mm}
	let P be a pattern and i be a global index.
	
	\Fn{$\FastMatching(P)$}{
		$U \ot \Root$;
		$i \ot 1$\;
		\While{$i \le |P|$}{
			\If{$\Child(U,P[i]) \ne \NULL$}{
				$U \ot \FastDecompact(U,\Child(U,P[i]))$\;
				\lIf{$U = \NULL$}{
					$\Ret$ $\False$}
			}
			\lElse{
				$\Ret$ $\False$}
		}
		$\Ret$ $\True$\;
	}
	
	\Fn{$\FastDecompact(U,V)$}{
		\While{$U \ne V$}{
			\If{$\Child(U,P[i]) \ne \NULL$}{
				\If{$\Plus(\Child(U,P[i])) = \False$}{
					$U \ot \Child(U,P[i])$\;
					$i \ot i+1$\;
				}
				\lElse{
					$U = \FastDecompact(\FLink(U),\FLink(\Child(U,P[i])))$}
				\lIf{$i > |P|$}{
					$\Ret$ $V$}
			}
			\lElse{
				$\Ret$ $\NULL$}
		}
		$\Ret$ $V$\;
	}
\end{algorithm2e}

\end{document}